\documentclass[12pt,onecolumn, draftcls, journal]{IEEEtran}

\ifCLASSINFOpdf
  \usepackage[pdftex]{graphicx}
  \DeclareGraphicsExtensions{.pdf,.jpeg,.png}
\else
  \usepackage[dvips]{graphicx}
\fi
\usepackage{epstopdf}
\usepackage[compress,nospace]{cite}
\usepackage{amsthm}
\usepackage{stfloats}
\usepackage[cmex10]{amsmath}
\usepackage{amsmath}
\usepackage{siunitx}
\usepackage{amssymb}
\usepackage{wasysym}
\usepackage{algorithm}
\usepackage{algorithmic}
\usepackage[colorlinks,linkcolor=blue,anchorcolor=blue,citecolor=blue,bookmarks=blue]{hyperref}
\usepackage{array}
\usepackage{url}
\usepackage{color}
\usepackage{subfigure}
\usepackage{verbatim}
\usepackage{multirow}
\usepackage{booktabs}
\usepackage{mathrsfs}
\usepackage{amsmath,amssymb,amsfonts}
\newtheorem{theorem}{Theorem}
\newtheorem{lemma}{Lemma}

\def\BibTeX{{\rm B\kern-.05em{\sc i\kern-.025em b}\kern-.08em
		T\kern-.1667em\lower.7ex\hbox{E}\kern-.125emX}}

\begin{document}
	
\title{Power Allocation in Multi-User Cellular Networks: Deep Reinforcement Learning Approaches}

\author{Fan~Meng, Peng~Chen,~\IEEEmembership{Member,~IEEE}, Lenan~Wu\\ and Julian Cheng,~\IEEEmembership{Senior Member,~IEEE}
	\thanks{Fan~Meng, and Lenan~Wu are with the School of Information Science and
	Engineering, Southeast University, Nanjing 210096, China (e-mail: mengxiaomaomao@outlook.com, wuln@seu.edu.cn).}
	\thanks{Peng Chen is with the State Key Laboratory of Millimeter Waves, Southeast University, Nanjing 210096, China (e-mail: chenpengseu@seu.edu.cn).}
	\thanks{J. Cheng is with School of Engineering, The University of British Columbia, Kelowna, V1V 1V7, BC, Canada (e-mail: julian.cheng@ubc.ca).}
}

\maketitle

\begin{abstract}
	
The model-based power allocation algorithm has been investigated for decades, but it requires the mathematical models to be analytically tractable and it usually has high computational complexity. Recently, the data-driven model-free machine learning enabled approaches are being rapidly developed to obtain near-optimal performance with affordable computational complexity, and deep reinforcement learning (DRL) is regarded as of great potential for future intelligent networks. In this paper, the DRL approaches are considered for power control in multi-user wireless communication cellular networks. Considering the cross-cell cooperation, the off-line/on-line centralized training and the distributed execution, we present a mathematical analysis for the DRL-based top-level design. The concrete DRL design is further developed based on this foundation, and policy-based \emph{REINFORCE}, value-based deep Q learning (DQL), actor-critic deep deterministic policy gradient (DDPG) algorithms are proposed. Simulation results show that the proposed data-driven approaches outperform the state-of-art model-based methods on sum-rate performance, with good generalization power and faster processing speed. Furthermore, the proposed DDPG outperforms the \emph{REINFORCE} and DQL in terms of both sum-rate performance and robustness, and can be incorporated into existing resource allocation schemes due to its generality. 

\end{abstract}


\section{Introduction}\label{sec:intro}

Wireless data transmission has experienced tremendous growth in past years and will continue to grow in the future. When large numbers of terminals such as mobile phones and wearable devices are connected to the networks, the density of access point (AP) will have to be increased. Dense deployment of small cells such as pico-cells, femto-cells, has become the most effective solution to accommodate the critical demand for spectrum~\cite{6815898}. With denser APs and smaller cells, the whole communication network is flooded with wireless signals, and thus the intra-cell and inter-cell interference problems are severe~\cite{Zhang2015Coexistence}. Therefore, power allocation and interference management are crucial and challenging~\cite{Luo2008Dynamic},~\cite{Boccardi2013Five}.

Massive model-oriented algorithms have been developed to cope with interference management~\cite{Shen2018Fractional, Shi2011An, Chiang2008Power, Zhang2011Weighted, Yu2013Multicell}, and the existing studies mainly focus on sub-optimal or heuristic algorithms, whose performance gaps to the optimal solution are typically difficult to quantify. Besides, the mathematical models are usually assumed to be analytically tractable, but these models are not always accurate because both hardware and channel imperfections can exist in practical communication environments. When considering specific hardware components and realistic transmission scenarios, such as low-resolution A/D, nonlinear amplifier and user distribution, the signal processing techniques with model-driven tools are challenging to be developed. Moreover, the computational complexity of these algorithms is high and thus concrete implementation becomes impractical. Meanwhile, machine learning (ML)~\cite{Bishop2006Pattern} algorithms are potentially useful techniques for future intelligent wireless communications. These methods are usually model-free and data-driven\cite{Oshea2017An},~\cite{Wang2017Deep}, and the solutions are obtained through data learning instead of model-oriented analysis and design. 

Two main branches of ML are supervised learning and reinforcement learning (RL). With available training input/output pairs, the supervised learning method is simple but efficient especially for classification tasks such as modulation recognition~\cite{8454504} and signal detection~\cite{8052521},~\cite{8283347}. However, the correct output data sets or optimal guidance solutions can be difficult to obtain. Meanwhile, the RL~\cite{Sutton1998Reinforcement} has been developed as a goal-oriented algorithm, aiming to learn a better policy through exploration of uncharted territory and exploitation of current knowledge. The RL concerns with how agents ought to take actions in an environment so as to maximize some notion of cumulative reward, and the environment is typically formulated as a Markov decision process (MDP)~\cite{Busoniu2010Reinforcement}. Therefore, many RL algorithms~\cite{Sutton1998Reinforcement} have been developed using dynamic programming (DP) techniques. In classic RL, a value function or a policy is stored in a tabular form, which leads to the curse of dimensionality and the lack of generalization. To compromise generality and efficiency, function approximation is proposed to replace the table, and it can be realized by a neural network (NN) or deep NN (DNN)~\cite{Goodfellow2016Deep}. Combining RL with DNN, the deep RL (DRL) is created and widely investigated, and it has achieved stunning performance in a number of noted projects~\cite{DBLP:journals/corr/Li17b} such as the game of Go~\cite{Silver2017Mastering} and Atari video games~\cite{Mnih2015Human}.

The DRL algorithms can be categorized into three groups~\cite{DBLP:journals/corr/Li17b}: value-based, policy-based and actor-critic methods. The value-based DRL algorithm derives optimal action by the action-state value function, and the most widely-used algorithms include deep Q learning (DQL) and Sarsa. As for the policy-based algorithm such as \emph{REINFORCE}, a stochastic policy is directly generated. Both of these two methods have the following defects in general:
\begin{enumerate}
	\item Value-based: The action space must be discrete, which introduces quantization error for tasks with continuous action space. The output dimension increases exponentially for multi-action issues or joint optimizations.
	\item Policy-based: It is difficult to achieve a balance between exploration and exploitation, and the algorithm usually converges with a suboptimal solution. The variance of estimated gradient is high. In addition, the action space is still discrete.
\end{enumerate}
The actor-critic algorithm is developed to overcome the aforementioned drawbacks as a hybrid of the value-based and policy-based methods. It consists of two components: an actor to generate policy and a critic to assess the policy. A better solution is learned through settling a multi-objective optimization problem, and updating the parameters of the actor and the critic alternatively.

In a communication system where multiple users share a common frequency band, the problem of choosing transmit power dynamically in response to physical channel conditions in order to maximize the downlink sum-rate with maximal power constraints is NP-hard~\cite{Luo2008Dynamic}. Two advanced model-based algorithms, namely fractional programming (FP)~\cite{Shen2018Fractional} and weighted minimum mean squared error (WMMSE)~\cite{Shi2011An} are regarded as benchmarks in the simulation comparisons. The supervised learning is studied in~\cite{8444648}, where a DNN is utilized to mimic the guidance algorithm, and accelerate the processing speed with acceptable performance loss. The ensemble of DNNs is also proposed to further improve the performance in~\cite{DBLP:DeepNeural}. As for the interference management/power allocation with DRL approaches, the current research work mainly concentrates on value-based methods. The QL or DQL is widely applied in various communication scenarios by a number of articles, such as Het-Nets~\cite{5700414, 5983301, 6965655, 8422864}, cellular networks~\cite{7997440},~\cite{DBLP:DRL} and V2V broadcasting~\cite{8450518}. To the best of the authors' knowledge, the classic policy-based approach has seldom been considered on this issue~\cite{4525466}. An actor-critic algorithm has been applied for power allocation~\cite{8100645}, where a Gaussian probability distribution is used to formulate a stochastic policy.

In this paper, we consider an interfering multiple-access channel (IMAC) scenario which is similar to~\cite{8444648}. We focus on the system-level optimization, and target at maximizing the overall sum-rate by inter-cell interference coordination. This is actually a static optimization problem, where the target is a multi-variate ordinary function. While the standard DRL tools are designed for the DP which can be settled recursively, a direct utilization of these tools to tackle the static optimization problem will suffer some performance degradation. In our previous work~\cite{Meng}, we verified through simulations that the present widely applied standard DQL algorithm suffers sum-rate performance degradation on power allocation. In this work, we explain the reasons for this degradation and revise the DRL algorithms eliminate such degradation, by developing theoretical analysis on the general DRL approaches to address the static optimization problem. On this theoretical basis, three more simplified but efficient algorithms, namely policy-based \emph{REINFORCE}, value-based DQL and actor-critic-based deep deterministic policy gradient (DDPG)~\cite{Lillicrap2015Continuous} are proposed. Simulation results show that the proposed DQL achieves the highest sum-rate performance when compared to the ones with standard DQL, and our DRL approaches also outperform the state-of-art model-based methods. The contributions of this manuscript are summarized as follows:
\begin{itemize}
	\item We develop mathematical analysis on proper application of general DRL algorithms to address the static optimization problems, and we consider dynamic power allocation in multi-user cellular networks. 
	\item The training procedure of the proposed DRL algorithm is centralized and the learned model is distributively executed. Both the off-line and on-line training are introduced, and an environment tracking mechanism is proposed to dynamically control the on-line learning.
	\item The logarithmic representation of channel gain and power is used to settle numerical problem in DNNs and improve training efficiency. Besides, a sorting preprocessing technique is proposed to accommodate varying user densities and reduce computation load.
	\item On the basis of proposed general DRL on static optimization, the concrete DRL design is further introduced and we propose three novel algorithms, namely \emph{REINFORCE}, DQL and DDPG, which are respectively policy-based, value-based and actor-critic-based. Contrast simulations on sum-rate performance, generalization ability and computation complexity are also demonstrated.
\end{itemize}

The remainder of this paper is organized as follows. Section~\ref{sec:system} outlines the power control problem in the wireless cellular network with IMAC. In Section~\ref{sec:DRL1}, the top-level DRL design for static optimization problem is analyzed and introduced. In Section~\ref{sec:DRL2} our proposed DRL approaches are presented in detail. Then, the DRL methods are compared along with benchmark algorithms in different scenarios, and the simulation results are demonstrated in Section~\ref{sec:sim}. Conclusions and discussion are given in Section~\ref{sec:con}. 

\section{System Model}\label{sec:system}

We investigate cross-cell dynamic power allocation in a wireless cellular network with IMAC. The network system is composed of $ N $ cells, and a base station (BS) with one transmitter is deployed at each cell center. Assuming shared frequency bands, $ K $ users are simultaneously served by the center BS.

\subsection{Problem Formulation}

At time slot $ t $, the independent channel gain between the BS $ n $ and the user $ k $ in cell $ j $ is denoted by $ g^t_{n,j,k} $, and can be presented as 
\begin{equation}\label{equ:g}
g^t_{n,j,k} = |h^t_{n,j,k}|^2\beta_{n,j,k}
\end{equation}
where $ |\cdot| $ is the absolute value operation; $ h^t_{n,j,k} $ is a complex Gaussian random variable with Rayleigh distributed magnitude; $ \beta_{n,j,k} $ is the large-scale fading component, taking both geometric attenuation and shadow fading into account, and it is assumed to be invariant over the time slot. According to the Jakes’ model~\cite{Bottomley1993Jakes}, the small-scale flat fading is modeled as a first-order complex Gauss-Markov process
\begin{equation}\label{equ:h}
h^t_{n,j,k} = \rho h^{t-1}_{n,j,k} + n^t_{n,j,k}
\end{equation}
where $ h^1_{n,j,k} \sim \mathcal{CN}(0, 1) $ and $ n^t_{n,j,k} \sim \mathcal{CN}(0, 1-\rho^{2}) $. The correlation $ \rho $ is determined by
\begin{equation}\label{equ:rho}
\rho = J_0(2\pi f_d T_s)
\end{equation}
where $ J_0(\cdot) $ is the first kind zero-order Bessel function, $ f_d $ is the maximum Doppler frequency, and $ T_s $ is the time interval between adjacent instants.

The downlink from the $ n $-th BS to the $ k $-th serving AP is denoted by $ \textup{dl}_{n, k} $. Supposing that the signals from different transmitters are independent of each other, the channels remain constant in each time slot. Then the signal-to-interference-plus-noise ratio (SINR) of $ \textup{dl}_{n, k} $ in time slot $ t $ can be formulated by
\begin{equation}\label{equ:sinr}
\boldsymbol{\gamma}^t_{n,k} = \frac{g^t_{n,n,k} p^t_{n,k}}{\sum_{k' \neq k} g^t_{n,n,k} p^t_{n,k'} + \sum_{n' \in D_n} g^t_{n',n,k} \sum_{j} p^t_{n',j} + \sigma^2}
\end{equation}
where $ D_n $ is the set of interference cells around the $ n $-th cell, $ p^t_{n,k} $ is the emitting power of the transmitter $ n $ to its receiver $ k $ at slot $ t $, and $ \sigma^2 $ denotes the additional noise power. The terms $ \sum_{k' \neq k} g^t_{n,n,k} p^t_{n,k'} $ and $ \sum_{n' \in D_n} g^t_{n',n,k} \sum_{j} p^t_{n',j} $ represent the intra-cell and inter-cell interference power, respectively. With normalized bandwidth, the downlink rate of $ \textup{dl}_{n, k} $ in time slot $ t $ is expressed as 
\begin{equation}\label{equ:C}
C^t_{n,k} = \log_2\left(1 + \boldsymbol{\gamma}^t_{n,k}\right).
\end{equation}
Under maximum power constraint of each transmitter, our goal is to find the optimum power, to maximize the sum-rate objective function. The optimization problem is given as
\begin{equation}\label{equ:opt1}
\begin{split}
& \max_{\boldsymbol{p}^t} \quad C(\boldsymbol{g}^t, \boldsymbol{p}^t)\\
& \textup{s.t.}\quad 0 \leq p^t_{n,k} \leq P_{\textup{max}}, \;\forall n,k
\end{split}
\end{equation}
where $ P_{\textup{max}} $ denotes the maximum emitting power; the power set $ \boldsymbol{p}^t $, channel gain set $ \boldsymbol{g}^t $, and sum-rate $ C(\boldsymbol{g}^t, \boldsymbol{p}^t) $ are respectively defined as
\begin{align}
& \boldsymbol{p}^t := \{p^t_{n,k} \mid \forall n,k\},\\
& \boldsymbol{g}^t := \{g^t_{n',n,k} \mid \forall n',n,k\},\\
& C(\boldsymbol{g}^t, \boldsymbol{p}^t) := \sum_{n, k} C^t_{n,k}.\label{equ:sumrate}
\end{align}
The problem \eqref{equ:opt1} is non-convex and NP-hard. As for the model-based methods, the performance gaps to the optimal solution are typically difficult to quantify, and also the practical implementation is restricted due to the high computational complexity. More importantly, the model-oriented approached cannot accommodate future heterogeneous service requirements and randomly evolving environments, and thus the data-driven DRL algorithms are discussed and studied in the following section.

\section{Deep Reinforcement Learning}\label{sec:DRL1}

\subsection{Problem Formulation}

A general MDP problem concerns about a single or multiple agents interacting with an environment. In each interaction, the agent takes action $ a $ by policy $ \pi $ with observed state $ s $, then receives a feedback reward $ r $ and a new state from the environment. The agent aims to find an optimal policy to maximize the cumulative reward over the continuous interactions, and the DRL algorithms are developed for such problems.

To facilitate the analysis, the discrete-time model-based MDP is considered, and the action and state spaces are assumed to be finite. The $ 4 $-tuple $ (S,A,P,R) $ is known, where the elements are
\begin{enumerate}
	\item $ S $, a finite set of states,
	\item $ A $, a finite set of actions,
	\item $ P^{a}_{s \to s'} = \Pr(s'|s, a) $ is the probability that action $ a $ in state $ s $ will lead to state $ s' $,
	\item $ R $, a finite set of immediate rewards, where element $ r^{a}_{s \to s'} $ denotes the reward obtained after transitioning from state $ s $ to state $ s' $, due to action $ a $.
\end{enumerate}
Under stochastic policy $ \pi $, the $ T $-step cumulative reward and $ \gamma $-discounted cumulative reward are considered as the state value function $ V $. With initial state $ s^1 $, they are defined as
\begin{align}
V^T_{\pi}(s^1) & := \mathbb{E}_{\pi}\left[\frac{1}{T} \sum_{t = 1}^{T} r^t \mid s^1\right]\label{equ:V1}
\end{align}
and
\begin{align}
V^\gamma_{\pi}(s^1) & := \lim_{T \to \infty}\mathbb{E}_{\pi}\left[\sum_{t = 1}^{T} \gamma^{t-1} r^t \mid s^1\right]\label{equ:V2}
\end{align}
where $ \gamma \in [0, 1) $ denotes a discount factor that trades off the importance of immediate and future rewards, and $ \mathbb{E}[\cdot] $ is the expectation operation. For an initial state-action pair $ (s^1, a^1) $, the state-action value functions, namely the Q functions are defined as
\begin{align}
Q^T_{\pi}(s^1) & := \mathbb{E}_{\pi}\left[\frac{1}{T} \sum_{t = 1}^{T} r^t \mid s^1, a^1\right]\label{equ:Q1}
\end{align}
and
\begin{align}
Q^{\gamma}_{\pi}(s^1) & := \lim_{T \to \infty} \mathbb{E}_{\pi}\left[\sum_{t = 1}^{T} \gamma^{t-1} r^{t} \mid s^1, a^1\right].\label{equ:Q2}
\end{align}

Starting from the perspective of MDP, the following conclusions are given when the environment satisfies certain conditions. 

\begin{theorem}\label{the:1}
	
When the environment transition is independent with action, and the current action is only related to the reward function of this instant, then the optimal policy for maximization of cumulative rewards is equivalent to a combination of single-step rewards.

\end{theorem}

\begin{proof}
	
First, we focus on \eqref{equ:V1} and it is expanded as
\begin{equation}\label{equ:V3}
\begin{split}
V^T_{\pi}(s^1) = & \sum_{a^1 \in A} \pi(a^1|s^1) \sum_{s^2 \in S} P^{a^1}_{s^1 \to s^2}\\
& \times \left(\frac{1}{T} r^{a^1}_{s^1 \to s^2} + \frac{T-1}{T} V^{T-1}_{\pi}(s^2)\right).
\end{split}
\end{equation}
The description of the assumed conditions can be mathematically formulated as
\begin{align}
& P^{a}_{s \to s'} = P_{s \to s'},\label{equ:cond1}\\
& r^{a}_{s \to s'} = r_{s}^{a}.\label{equ:cond2}
\end{align} 
Without loss of generality, for probability mass functions of policy $ \pi $ and state transitioning $ P $, clearly we have
\begin{align}
& \sum_{a \in A} \pi(a|s) = 1,\label{equ:cond3}\\
& \sum_{s' \in S} P_{s \to s'} = 1.\label{equ:cond4}
\end{align}
From \eqref{equ:cond1}, \eqref{equ:cond2}, \eqref{equ:cond3} and \eqref{equ:cond4}, the state value function \eqref{equ:V3} can be rewritten as
\begin{equation}\label{equ:V4}
\begin{split}
V^T_{\pi}(s^1) = & \sum_{a^1 \in A} \pi(a^1|s^1) \sum_{s^2 \in S} P_{s^1 \to s^2}\\
& \times \left(\frac{1}{T} r^{a^1}_{s^2} + \frac{T-1}{T} V^{T-1}_{\pi}(s^2)\right)\\
= & \frac{1}{T} \sum_{a^1 \in A} \pi(a^1|s^1) r^{a^1}_{s^1}\\
& + \frac{T-1}{T} \sum_{s^2 \in S} P_{s^1 \to s^2} V^{T-1}_{\pi}(s^2).
\end{split}
\end{equation}
The full unrolling of \eqref{equ:V4} is given as
\begin{equation}\label{equ:V6}
\begin{split}
V^T_{\pi}(s^1) = & \frac{1}{T} \sum_{a^1 \in A} \pi(a^1|s^1) r^{a^1}_{s^1}\\ 
& +\frac{1}{T} \sum_{t=2}^T \sum_{a^t \in A} \pi(a^t|s^t) \sum_{s^{t} \in S} \prod_{t'=1}^{t-1} P_{s^{t'} \to s^{t'+1}} r^{a^t}_{s^{t}}.
\end{split}
\end{equation} 
Since the state transfer is irrelevant to the action, and the state can be independently sampled as  $ \boldsymbol{s} = <s^1,\cdots,s^{T+1}> $.
\begin{lemma}\label{lem:1}
With state sequence $ \boldsymbol{s} = <s^1,\cdots,s^{T+1}> $, the maximization of \eqref{equ:V6} with respect to $ a^t,\forall t $ can be decomposed to the subproblem:
\begin{equation}\label{equ:V555}
\max_{a^t} V^T_{\pi}(\boldsymbol{s}) \iff \max_{a^t} r^{a^t}_{s^t}.
\end{equation}
\end{lemma}

\begin{proof}
\begin{equation}\label{equ:V55}
\begin{split}
\max_{a^t} V^T_{\pi}(\boldsymbol{s}) & = \max_{a^t} \frac{1}{T} \sum_{t' = 1}^T \sum_{a^{t'} \in A} \pi(a^{t'}|s^{t'}) r^{a^{t'}}_{s^{t'}}\\
& \iff \max_{a^t} \sum_{a^t \in A} \pi(a^t|s^t) r^{a^t}_{s^t}\\
& = \max_{a^t} r^{a^t}_{s^t}.\\
\end{split}
\end{equation}
\end{proof}

Obviously with Lemma~\ref{lem:1}, it can be proved that the maximization of \eqref{equ:V6} with respect to $ \{a^t|\forall t\} $ can be decomposed into $ T $ subproblems:
\begin{equation}\label{equ:V556}
\max_{\{a^t|\forall t\}} V^T_{\pi}(\boldsymbol{s}) \iff \Big\{\max_{a^t} r^{a^t}_{s^t}\Big|\forall t\Big\}.
\end{equation}
Besides, the equivalence proof of $ \gamma $-discounted cumulative reward is similar.

\end{proof}

Since the channel is modeled as a first-order Markov process, the environment satisfies the two conditions in Theorem~\ref{the:1}. Then we let $ r^{a^t}_{s^t} = C(\boldsymbol{g}^t, \boldsymbol{p}^t) $ and $ a^t = \boldsymbol{p}^t $ along with the constraints, the optimization problem of \eqref{equ:opt1} with DRL approach is equivalent to \eqref{equ:V55}.

Although the equivalence is mathematically proved, and it has no concern with the value of $ T $ or $ \gamma $. Several facts must be observed when the improper hyper-parameter is adopted. We take the value-based method as an example, and optimal Q function associated with Bellman equation is given as 
\begin{equation}\label{equ:bellman}
Q^*(s,a) = r^a_s + \gamma \max_{a'}Q(s',a').
\end{equation}
This function must be estimated precisely to achieve the optimal action. Here we list two issues caused by $ \gamma>0 $:
\begin{enumerate}
	\item The Q value is overestimated, and the bias is $ \gamma \max_{a'}Q(s',a') $. This effect actually has no or little influence on the final performance, since this deviation does not concern with action $ a $. 
	\item The variance of Q value $ \sigma^2_q $ becomes enlarged, and $ \sigma^2_q $ becomes larger as $ \gamma $ increases. During training, the additional noise on data can slow down the convergence speed, and also can deteriorate the performance of learned DNN.
\end{enumerate}
In~\cite{Meng}, we verified that an increasing $ \gamma $ has negative influence on the sum-rate performance of DQN in simulations, as shown in Fig.~\ref{fig:gamma}. \textbf{Therefore, we suggest using hyper-parameter $ \gamma = 0 $ or $ T = 1 $ in this specific scenario, and thus $ Q(s,a) = r^a_s $.} In the remainder of this article, we make adjustment to the standard DRL algorithms, and particularly claim that the Q function is equal to the reward function. The aforementioned analysis and discussion provide the design guidance for the next DRL.

\begin{figure}
	\centering
	\includegraphics[width=3.6in]{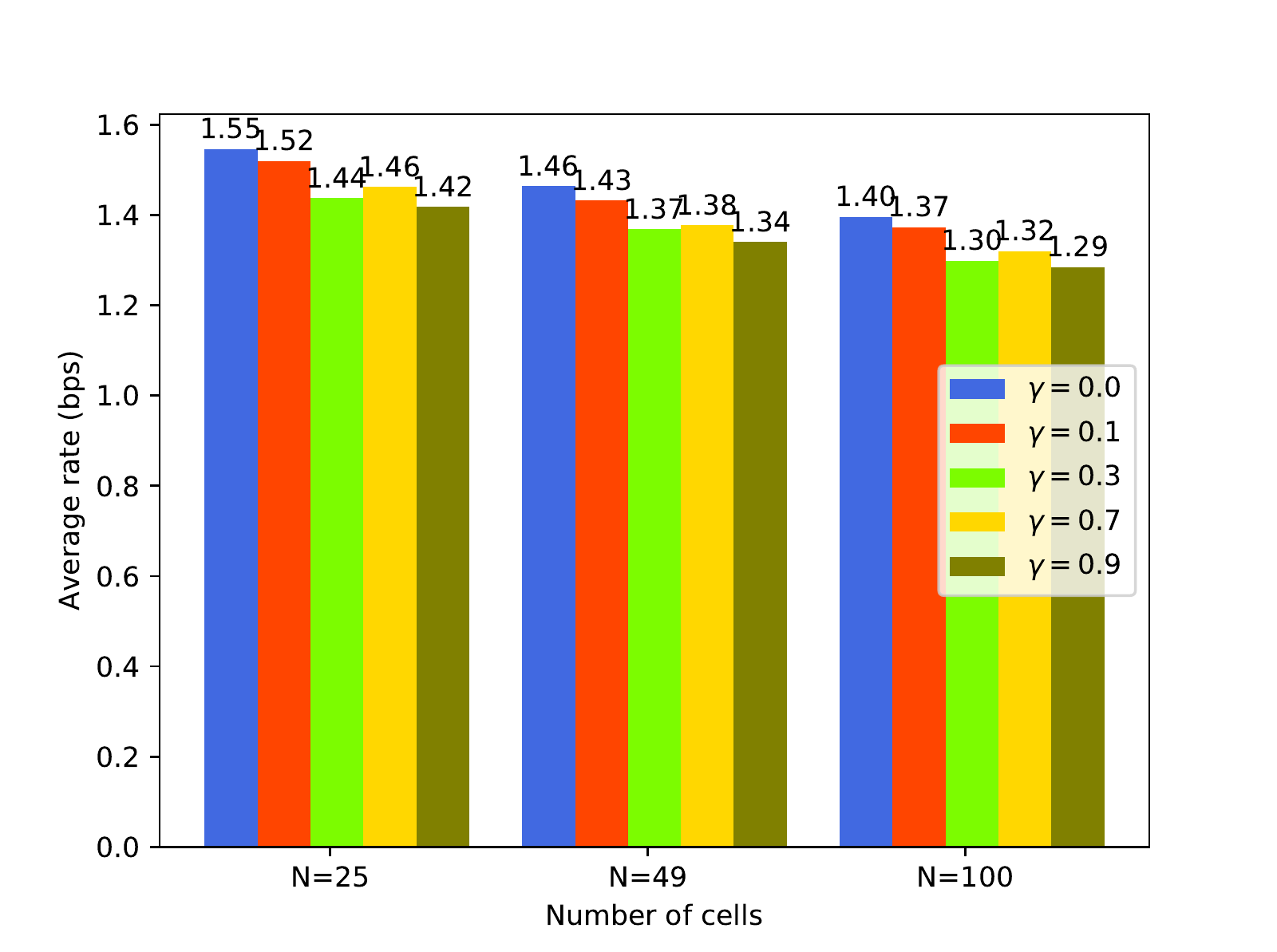}
	\caption{The average sum-rate versus cellular network scalability for trained DQNs with different $ \gamma $ values.}
	\label{fig:gamma}
\end{figure}

\subsection{Centralized Training \& Distributed Execution}\label{sec:cen-tra}

In \eqref{equ:opt1}, only a single center agent is trained and then implemented. Under this framework, the current local channel state information (CSI) is first estimated and transmitted to the center agent for further processing. The decisions of allocated powers are then broadcast to the corresponding transmitters and executed. However, several defects of the centralized framework with a massive number of cells must be observed:
\begin{enumerate}
	\item Space explosion: The cardinalities of DNN I/O is proportional to the cell number $ N $, and training such a DNN is difficult since the state-action space increases exponentially with the I/O dimensions. Additionally, exploration in high-dimensional space is inefficient, and thus the learning can be impractical.
	\item Delivery pressure: The center agent requires full CSI of the communication network in current time. When the cell number $ N $ is large and low-latency service is required, both the transmitting CSI to the agent and the broadcasting allocation scheme to each transmitter are challenging. 
\end{enumerate}
In~\cite{8466370}, a framework of centralized training and distributed execution was proposed to address these challenges. The power allocation scheme is decentralized, the transmitter of each link is regarded as an agent, and all agents in the communication network operate synchronously and distributively. Meanwhile, the agent $ n,k $ only partially consumes channel information $ \boldsymbol{g}^t_{n,k} $ and outputs its own power $ p^t_{n,k} $, where $ \boldsymbol{g}^t_{n,k} $ is defined as
\begin{equation}\label{equ:g_t}
\boldsymbol{g}^t_{n,k} = \{g^t_{n',n,k} \mid n'\in\{n,D_n\},\forall k\}.
\end{equation}
The multi-objective programming is established as 
\begin{equation}\label{equ:opt4}
\begin{split}
& \max_{\boldsymbol{p}^t} \quad \Big\{C(\boldsymbol{g}^t_{n,k}, p^t_{n,k})\mid\forall n,k\Big\} \\
& \textup{s.t.}\quad 0 \leq p^t_{n,k} \leq P_{\textup{max}}, \;\forall n,k.
\end{split}
\end{equation}
However, multi-agent training is still difficult, since it requires much more learning data, training time and DNN parameters. Besides, links in distinct areas are approximately identical since their characteristics are location-invariant and the network is large. To simplify this issue, all agents are treated as the same agent. Same policy is shared and it is learned with collected data from all links. Therefore, the training is centralized, and the execution is distributed. The detailed design of the DRL algorithms will be introduced in the following section.

\subsection{On-line Training}\label{sec:two-step}

In our previously proposed model-free two-step training framework~\cite{Meng}, the DNN is first pretrained off-line in simulated wireless communication scenarios. This procedure is to reduce the on-line training stress, due to the large data requirement for data-driven algorithm by nature. Second, with transfer learning, the off-line learned DNN can be deployed in real networks. However, it will suffer from the imperfections in real implementations, dynamic wireless channel environment and some unknown issues. Therefore, the agent must be trained on-line in the initial deployment, in order to adapt to actual unknown issues that cannot be simulated. To prevent a prolonged degradation of the system performance, parameter update of the DNN to accommodate the environment changes is also necessary. 

One simple but brute-force approach is to use continuous regular training, which leads to a great waste of network performance and computation resources. On-line training is costly for several reasons. First, interaction with the real environment is required, and this exploration ruins the sum-rate performance of communication system to some extent. Second, the training requires high performance computing to reduce time cost, while the hardware is expensive and power-hungry. On the one hand, training is unnecessary when the environment fluctuation is negligible, but on the other hand this method cannot timely respond to the outburst. 

Therefore, we propose an environment tracking mechanism as an efficient approach to dynamically control the agent training. For DRL algorithms, the shift of environment indicates that the reward function $ R $ is changed, and thus the policy $ \pi $ or Q function must be adjusted correspondingly to avoid performance degradation. Hence, the Q value needs to approximate the reward value $ r $ as accurately as possible. We define the normalized critic loss $ l^t_c $ as
\begin{equation}\label{equ:critic1}
l^t_c = \frac{1}{2T_l} \sum_{t'= t-T_l+1}^t \left(1 - \frac{Q(s^{t'}, a^{t'} ; \boldsymbol{\theta})}{r^{t'}}\right)^2
\end{equation}
where $ \boldsymbol{\theta} $ denotes the DNN parameter; $ T_l $ is the observation window; $ l^t_c $ is an index to evaluate the accuracy of Q function approximation to the actual environment. Once $ l^t_c $ exceeds some fixed threshold $ l_{\max} $, the training of DNN is initiated to track the current environment; otherwise, the learning procedure is omitted. The introduction of tracking mechanism achieves a balance between performance and efficiency. With on-line training, the DRL is model-free and data-driven in a true sense.

\section{DRL Algorithm Design}\label{sec:DRL2}

\subsection{Concrete DRL Design}

In the previous section we discuss the DRL on a macro-level, and concrete design of several DRL algorithms namely \emph{REINFORCE}, DQL and DDPG, is introduced in this subsection. First, the descriptions of state, reward and action are given, as an expansion of Section~\ref{sec:cen-tra}. 

\subsubsection{State}

The selection of environment information is significant, and obviously current partial CSI $ \boldsymbol{g}^t_{n,k} $ is the most critical feature. It is inappropriate to directly use $ \boldsymbol{g}^t_{n,k} $ as DNN input due to numerical issues. In~\cite{Meng}, a logarithmic normalized expression of $ \boldsymbol{g}^t_{n,k} $ is proposed, and it is given as
\begin{equation}\label{equ:gamma_g}
\boldsymbol{\Gamma}_{n,k}^t := \frac{1}{g^t_{n,k,k}} \boldsymbol{g}^t_{n,k} \otimes \boldsymbol{1}_K
\end{equation}
where $ \otimes $ is the Kronecker product, and $ \boldsymbol{1}_K $ is a vector filled with $ K $ ones. The channel amplitudes elements in $ \boldsymbol{g}^t_{n,k} $ are normalized by the downlink $ \textup{dl}_{n, k} $, and the logarithmic representation is preferred since that amplitudes often vary by orders of magnitude. The cardinality of $ \boldsymbol{\Gamma}_{n,k}^t $ is $ (|D_n|+1)K $, and it changes with varying AP densities. First, we define the sorting function
\begin{equation}\label{equ:sort1}
\widetilde{x}, i := sort(x, y)
\end{equation}
where set $ x $ is sorted in decreasing order, and the first $ y $ elements are selected as the new set $ \widetilde{x} $. The indices of the chosen components are donated by $ i $. To further reduce the input dimension and accommodate different AP densities, the new set $ \widetilde{\boldsymbol{\Gamma}}_{n,k}^t $ and its indices $ I_{n,k}^t $ are obtained by \eqref{equ:sort1} with $ x = \boldsymbol{\Gamma}_{n,k}^t $ and $ y = I_c $, where $ I_c $ is a constant.

The channel is modeled as a Markov process and correlated in the time domain, and thus the last solutions can provide a better initialization for this moment's solve and interference information. In correspondence to $ \widetilde{\boldsymbol{\Gamma}}_{n,k}^t $, the last power set $ \widetilde{\boldsymbol{p}}^{t-1}_{n,k} $ is defined as 
\begin{equation}\label{equ:p_t}
\widetilde{\boldsymbol{p}}^{t-1}_{n,k} := \{p^{t-1}_{n,k}\mid (n,k) \in I_{n,k}^t\}.
\end{equation}
The irrelevant or weak-correlated input elements consume more computational resources and even lead to performance degradation, but some auxiliary information can improve the sum-rate performance of DNN. Similar to \eqref{equ:p_t}, the assisted feature is given by
\begin{equation}\label{equ:c_t}
\widetilde{\boldsymbol{C}}^{t-1}_{n,k} := \{C^{t-1}_{n,k}\mid (n,k) \in I_{n,k}^t\}.
\end{equation}
Two types of feature $ \mathrm{f} $ are considered, and they are written as
\begin{align}
\mathrm{f}_1 & := \{\widetilde{\boldsymbol{\Gamma}}_{n,k}^t, \widetilde{\boldsymbol{p}}^{t-1}_{n,k}\},\label{equ:f1}\\
\mathrm{f}_2 & := \{\widetilde{\boldsymbol{\Gamma}}_{n,k}^t, \widetilde{\boldsymbol{p}}^{t-1}_{n,k}, \widetilde{\boldsymbol{C}}^{t-1}_{n,k}\}.\label{equ:f2}
\end{align}
The partially observed state $ s $ for DRL algorithms can be $ \mathrm{f}_1 $ or $ \mathrm{f}_2 $, and their performance will be compared in the simulation section. Moreover, the cardinalities of state $ |S| $, i.e., the input dimensions, are $ 2I_c $ and $ 3I_c $.

\subsubsection{Reward}

According to our investigation, there is few work on the strict design criteria of reward function due to the problem complexity. In general, the reward function is elaborately designed to improve the agent's transmitting rate and also to mitigate its interference to neighbouring links~\cite{5983301, 6965655, 8422864, 7997440, DBLP:DRL, 8450518}. In our previous work, we use averaged sum-rate \eqref{equ:sumrate} as the reward, and it follows that the sum of all rewards is equal to the network sum-rate. However, rates from remote cells is introduced, and they have little relationship with decision of action $ p_{n,k}^{t} $. These irrelevant elements enlarge the variance of reward function, and thus the DNN becomes hard to train when the network becomes large. Therefore, the localized reward function is proposed as
\begin{equation}\label{equ:reward}
r_{n,k}^t := C_{n,k}^t + \alpha \left(\sum_{n,k'\neq k} C_{n,k'}^t + \sum_{n' \in D_n,j} C_{n',j}^t\right)
\end{equation}
where $ \alpha \in \mathbb{R}^+ $ is a weight coefficient of interference effect, and $ \mathbb{R}^+ $ denotes the positive real scalar. The sum of local rewards is proportional to the sum-rate 
\begin{equation}\label{equ:propotion}
\sum_{n,k} r_{n,k}^t \varpropto C(\boldsymbol{g}^t, \boldsymbol{p}^t)
\end{equation}
when the cell number $ N $ is sufficient large.

\subsubsection{Action}

The downlink power is a non-negative continuous scalar, and is limited by the maximum power $ P_{\textup{max}} $. Since that the action space must be finite for certain algorithms such as DQL and \emph{REINFORCE}, the possible emitting power is quantized in $ |A| $ levels. The allowed power set is given as
\begin{equation}\label{equ:p_set}
A := \left\{0, \left\{P_{\textup{min}}\left(\frac{P_{\textup{max}}}{P_{\textup{min}}}\right)^{\frac{i}{|A|-2}}\Big| i=0,\cdots,|A|-2\right\} \right\}
\end{equation}
where $ P_{\textup{min}} $ is the non-zero minimum emitting power. Discretization of continuous variable results in quantization error. Meanwhile, the actor of DDPG directly outputs deterministic action $ a = p_{n,k}^t $, and this constrained continuous scalar is generated by a scaled sigmoid function:
\begin{equation}\label{equ:ddpg_output}
a := P_{\textup{max}} \cdot \frac{1}{1 + \exp(-x)}
\end{equation}
where $ x $ is the pre-activation output. Except for elimination of quantization error, DDPG has great potential on multi-action task. For example, We take a task with action number $ N_A $ for example, the output dimension of DDPG $ |A| = N_A $. While for both DQL and \emph{REINFORCE}, we have $ |A| = \prod_i^{N_A} |A_i| $. Since the action space increases exponentially, the application of multi-action tasks with such algorithms is impractical. 

\subsubsection{Experience Replay}

The concept of ``Experience Replay" is proposed to deal with the problem: the data is correlated and non-stationary distributed in MDPs, while the training samples for DNN are suggested to be independently and identically distributed (I.I.D.). In our investigated problem, the data correlation in time domain is not strong, and this technique is optional.

\begin{figure}
	\centering
	\includegraphics[width=3.6in]{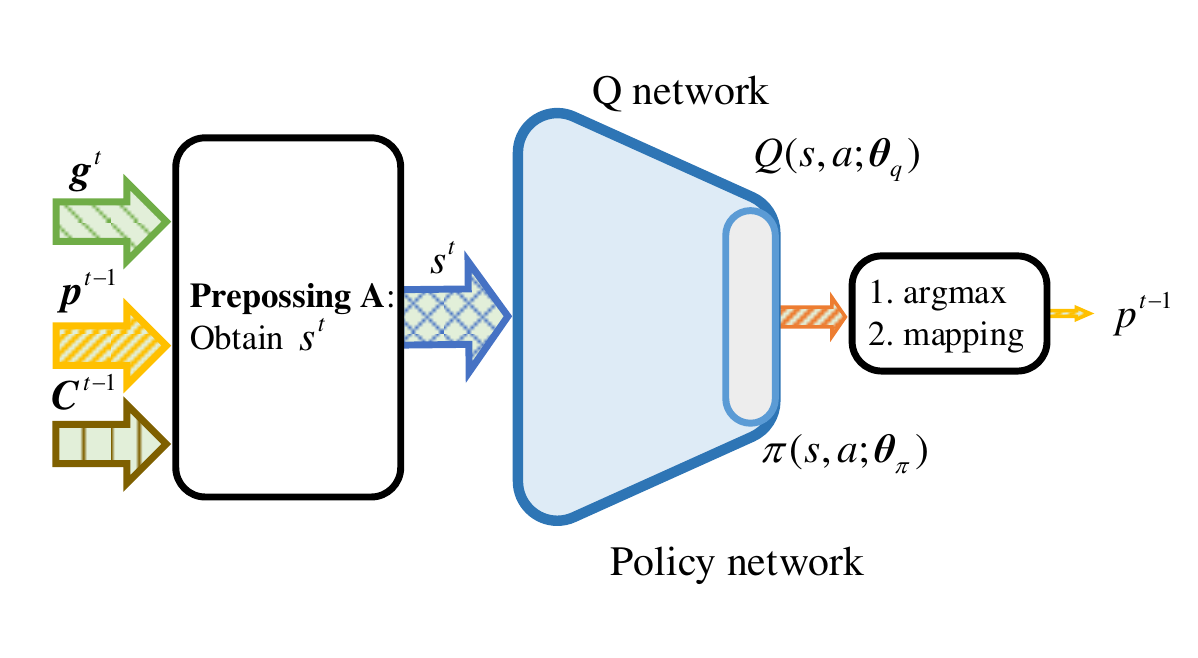}
	\caption{An illustration of data flow graph with \emph{REINFORCE} and DQL (feature $ \textup{f}_2 $).}
	\label{fig:network1}
\end{figure}

\subsection{Policy-based: \emph{REINFORCE}}

The \emph{REINFORCE} is derived as a Monte-Carlo policy-gradient learning algorithm~\cite{Sutton2000},\cite{ DBLP:journals/corr/ThomasB17}. Policy-based algorithms directly generate stochastic policy $ \pi $ instead of indirect Q valuation, and $ \pi $ is parameterized by a policy network $ \pi(a|s;\boldsymbol{\theta}_{\pi}) $ with parameter $ \boldsymbol{\theta}_{\pi} $, as shown in Fig.~\ref{fig:network1}. The overall strategy of stochastic gradient ascent requires a way to obtain samples such that the expectation of sample gradient is proportional to the actual gradient of the performance measure as a function of the parameter. The goal of \emph{REINFORCE} is to maximize expected rewards under policy $ \pi $:
\begin{equation}\label{equ:opt-policy}
\boldsymbol{\theta}_{\pi}^* = \arg \max_{\boldsymbol{\theta}_{\pi}} \mathbb{E}_{\pi} \left[\sum_a \pi(a|s;\boldsymbol{\theta}_{\pi})r^a_s\right]
\end{equation}
where $ \pi(a|s;\boldsymbol{\theta}_{\pi}) $ denotes the policy network, and $ \boldsymbol{\theta}_{\pi} $ is its parameter. The gradient of \eqref{equ:opt-policy} with Monte-Carlo sampling is presented as
\begin{equation}\label{equ:opt-policy-delta}
\nabla {\boldsymbol{\theta}_{\pi}} = \mathbb{E}_{\pi}  \big[\nabla_{\boldsymbol{\theta}_{\pi}} \ln \pi(a|s;\boldsymbol{\theta}_{\pi}) r^a_s \mid_{s=s^t,a=a^t} \big]
\end{equation}
where $ \nabla $ is the gradient operation. The complete deduction is presented in~\cite{Sutton1998Reinforcement}. Since the policy network $ \pi(a|s;\boldsymbol{\theta}_{\pi}) $ directly generates stochastic policy, the optimal action $ a^* $ is selected with the maximum probability:
\begin{equation}\label{equ:opt-policy-a}
a^* = \arg \max_a \pi(a|s;\boldsymbol{\theta}_{\pi})
\end{equation}
and the optimal action value is obtained by a mapping table. Besides, the random action is selected following $ \pi(a|s;\boldsymbol{\theta}_{\pi}) $ in exploration.

In practical training, the algorithm is susceptible to reward scaling. We can alleviate this dependency by whitening the rewards before computing the gradients, and the normalization of reward $ \tilde{r} $ is given as
\begin{equation}\label{equ:opt-policy-normalization}
\tilde{r} = \frac{r - \mu_{r}}{\sigma_{r}}
\end{equation}
where $ \mu_{r} $ and $ \sigma_{r} $ are the mean value and standard deviation of reward $ r $, respectively. The proposed \emph{REINFORCE} algorithm is stated in Algorithm.~\ref{alg:REINFORCE}.

\begin{algorithm}
	\caption{\emph{REINFORCE} algorithm.}
	\begin{algorithmic}[1]
		\STATE \emph{Input:} Episode times $ N_e $, exploration times $ T $, learning rate $ \eta_\pi $.
		\STATE \emph{Initialization:} Initialize policy network $ \pi(a|s;\boldsymbol{\theta}_{\pi}) $ with random parameter $ \boldsymbol{\theta}_{\pi} $.
		\FOR{$ k = 1 $ to $ N_e $}
		\STATE Receive initial state $ s^1 $. 
		\FOR{$ t = 1 $ to $ T $}
		\STATE Select action $ a^t $ following $ \pi(a|s;\boldsymbol{\theta}_{\pi}) $.
		\STATE Execute action $ a^t $, achieve reward $ r^t $ and observe new state $ s^{t+1} $.
		\STATE Calculate $ \tilde{r} $ by \eqref{equ:opt-policy-normalization}.
		\STATE Calculate gradient $ \nabla \boldsymbol{\theta}_{\pi} $ by \eqref{equ:opt-policy-delta}, and update parameter along positive gradient direction $ \boldsymbol{\theta}_{\pi} \gets \boldsymbol{\theta}_{\pi} + \eta_\pi \nabla \boldsymbol{\theta}_{\pi} $.
		\STATE $ s^t \gets s^{t+1} $.
		\ENDFOR
		\ENDFOR
		\STATE \emph{Output:} Learned policy network $ \pi(a|s;\boldsymbol{\theta}_{\pi}) $.
	\end{algorithmic}
	\label{alg:REINFORCE}
\end{algorithm}

\subsection{Value-based: DQL}

DQL is one of the most popular value-based off-policy DRL algorithms. As shown in Fig.~\ref{fig:network1}, the topology of DQL and \emph{REINFORCE} are the same, and the values are estimated by a DQN $ Q(s,a;\boldsymbol{\theta}_q) $, where $ \boldsymbol{\theta}_q $ denotes the parameter. The selection of a good action is based upon accurate estimation, and thus DQL is aimed to search for optimal parameter $ \boldsymbol{\theta}^*_q $ to minimize the $ \ell_2 $ loss:
\begin{equation}\label{equ:opt-value}
\boldsymbol{\theta}_q^* = \arg \min_{\boldsymbol{\theta}_q} \frac{1}{2} \left(Q(s,a;\boldsymbol{\theta}_q) - r^a_s\right)^2.
\end{equation}
The gradient with respect to $ \boldsymbol{\theta}_q $ is given as
\begin{equation}\label{equ:opt-value-delta}
\nabla {\boldsymbol{\theta}_q} = \left(Q(s,a;\boldsymbol{\theta}_q) - r^a_s\right) \nabla_{\boldsymbol{\theta}_q} Q(s,a;\boldsymbol{\theta}_q).
\end{equation}
The optimal action $ a^* $ is selected to maximize the Q value, and it is given by
\begin{equation}\label{equ:opt-value-a}
a^* = \arg \max_a Q(s,a;\boldsymbol{\theta}_q).
\end{equation}
During training, a dynamic $ \varepsilon $-greedy policy is adopted to control the exploration probability, and $ \varepsilon_k $ is defined as
\begin{equation}\label{equ:exp-value}
\varepsilon_k := \varepsilon_1 + \frac{k-1}{N_e-1}(\varepsilon_{N_e} - \varepsilon_1), k = 1,\cdots, N_e
\end{equation}
where $ N_e $ denotes the episode times,  $ \varepsilon_1 $ and $ \varepsilon_{N_e} $ are initial and final exploration probabilities, respectively. Detailed description of our DQL algorithm is presented in Algorithm.~\ref{alg:DQL}.

\begin{algorithm}
	\caption{DQL algorithm.}
	\begin{algorithmic}[1]
		\STATE \emph{Input:} Episode times $ N_e $, exploration times $ T $, learning rate $ \eta_q $, initial and final exploration probability $ \varepsilon_1 $, $ \varepsilon_{N_e} $.
		\STATE \emph{Initialization:} Initialize DQN $ Q(s,a;\boldsymbol{\theta}_q) $ with random parameter $ \boldsymbol{\theta}_q $.
		\FOR{$ k = 1 $ to $ N_e $}
			\STATE Update $ \varepsilon_k $ by \eqref{equ:exp-value}.
			\STATE Receive initial state $ s^1 $.
			\FOR{$ t = 1 $ to $ T $}
				\IF{$ \textup{rand()} < \varepsilon_k $} 
					\STATE Randomly select action $ a^t \in A $ with uniform probability. 
				\ELSE 
					\STATE Select action $ a^t $ by \eqref{equ:opt-value-a}.
				\ENDIF 
				\STATE Execute action $ a^t $, achieve reward $ r^t $ and observe new state $ s^{t+1} $.	
				\STATE Calculate gradient $ \nabla {\boldsymbol{\theta}_q} $ by \eqref{equ:opt-value-delta}, and update parameter along negative gradient direction: $ \boldsymbol{\theta}_q \gets \boldsymbol{\theta}_q - \eta_q \nabla {\boldsymbol{\theta}_q} $.
				\STATE $ s^t \gets s^{t+1} $.
			\ENDFOR
		\ENDFOR
		\STATE \emph{Output:} Learned DQN $ Q(s,a;\boldsymbol{\theta}_q) $.
	\end{algorithmic}
	\label{alg:DQL}
\end{algorithm}

\begin{figure}
	\centering
	\includegraphics[width=3.6in]{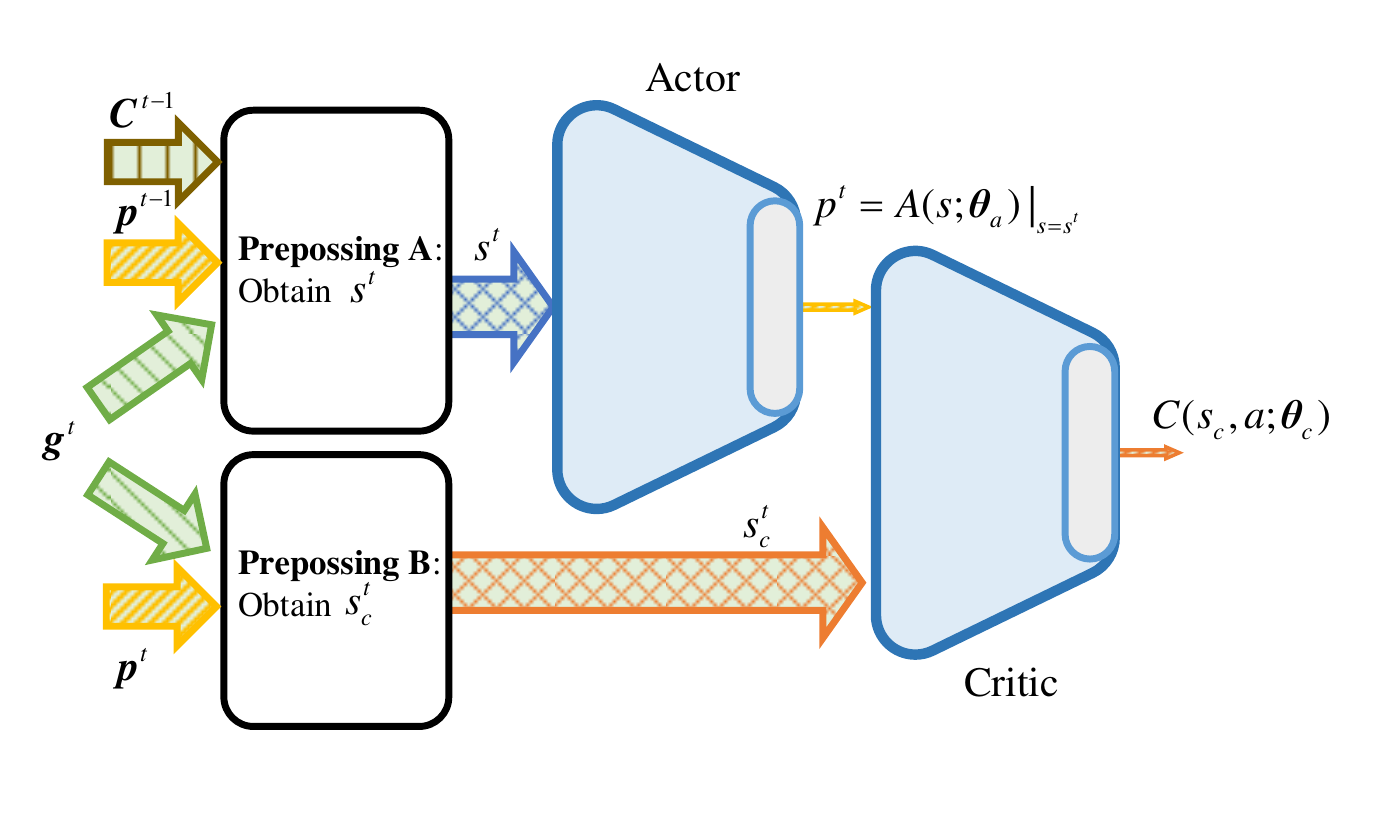}
	\caption{An illustration of data flow graph with DDPG (feature $ \textup{f}_2 $).}
	\label{fig:network2}
\end{figure}

\subsection{Actor-Critic: DDPG}

DDPG is presented as an actor-critic, model-free algorithm based on the deterministic policy gradient that can operate over continuous action spaces. As shown in Fig.~\ref{fig:network2}, an actor generates deterministic action $ a $ with observation $ s $ by a mapping network $ A(s;\boldsymbol{\theta}_a) $, where $ \boldsymbol{\theta}_a $ denotes the actor parameter. The critic predicts the Q value with an action-state pair through a critic network $ C(s_c,a;\boldsymbol{\theta}_c) $, where $ \boldsymbol{\theta}_c $ denotes the critic parameter and $ s_c $ is the critic state. The critic and actor work cooperatively, and the optimal deterministic policy is achieved by solving the following joint optimization problem:
\begin{align}
\boldsymbol{\theta}_a^* & = \arg \max_{\boldsymbol{\theta}_a} C(s_c,a;\boldsymbol{\theta}_c)\mid_{a = A(s;\boldsymbol{\theta}_a)},\label{equ:opt-ddpg-max}\\
\boldsymbol{\theta}_c^* & = \arg \min_{\boldsymbol{\theta}_c} \frac{1}{2} \left(C(s_c,a;\boldsymbol{\theta}_c)\mid_{a = A(s;\boldsymbol{\theta}_a)} - r^a_s\right)^2.\label{equ:opt-ddpg-min}
\end{align}
The actor strives to maximize the evaluation from critic, and the critic aims to make assessment precisely. Both the actor and critic are differentiable, and using chain rule their gradients are given as
\begin{align}
\nabla {\boldsymbol{\theta}_a} & = \nabla_{a} C(s_c,a;\boldsymbol{\theta}_c)\mid_{a = A(s;\boldsymbol{\theta}_a)} \nabla_{\boldsymbol{\theta}_a} A(s;\boldsymbol{\theta}_a),\label{equ:opt-ddpg-delta1}\\
\nabla {\boldsymbol{\theta}_c} & = \left(C(s_c,a;\boldsymbol{\theta}_c) - r^a_s\right) \nabla_{\boldsymbol{\theta}_c} C(s_c,a;\boldsymbol{\theta}_c)\mid_{a = A(s;\boldsymbol{\theta}_a)}.\label{equ:opt-ddpg-delta2}
\end{align}
The deterministic action is directly obtained by the actor:
\begin{equation}\label{equ:opt-ddpg-a}
a^* = A(s;\boldsymbol{\theta}_a).
\end{equation}
Similar to the dynamic $ \varepsilon $-greedy policy, the exploration action in episode $ k $ is defined as 
\begin{equation}\label{equ:opt-ddpg-noise}
a := \big\lbrack A(s;\boldsymbol{\theta}_a) + n^k\big\rbrack_0^{P_{\max}}
\end{equation}
where $ n^k $ is an additional noise and follows uniform distribution:
\begin{equation}\label{equ:noise}
n^k \sim \mathcal{U}(-\frac{P_{\max}}{k}, \frac{P_{\max}}{k})
\end{equation}
and action $ a $ is bounded by the interval $ [0, P_{\max}] $.

The critic $ C(s_c,a;\boldsymbol{\theta}_c) $ can be regarded as an auxiliary network to transfer gradient in learning, and it is needless in further testing. The $ C(s_c,a;\boldsymbol{\theta}_c) $ must be differentiable, but not necessarily trainable. The critic is model-based in this approach, since the evaluating rules are available with \eqref{equ:sinr}, \eqref{equ:C} and \eqref{equ:reward} in off-line training. 

However, the model-based actor is confirmed and infeasible to accommodate the unknown issues in on-line training. Meanwhile, complex reward function is difficult to be approximated accurately with pure NN parameters. Therefore, a semi-model-free critic is suggested, with utilization of both priori knowledge and flexibility of NN. Similar to the preprocessing of $ s $, the state for critic $ s_c = \widetilde{\boldsymbol{C}}_{n,k}^t $ is obtained by \eqref{equ:sinr}, \eqref{equ:C}, \eqref{equ:sort1} and \eqref{equ:c_t} with $ x = \{C_{n,k}^t|\forall n,k\} $ and $ y = I_c $. The detailed DDPG algorithm is introduced in Algorithm.~\ref{alg:DDPG}.

The policy gradient algorithm is developed with stochastic policy $ \pi(a|s) $, but sampling in continuous or high-dimensional action space is inefficient. The deterministic policy gradient is proposed to overcome this problem. On the other hand, in contrast with value-based DQL, the critic $ C(s,a;\boldsymbol{\theta}_c) $ and Q value estimator $ Q(s,a;\boldsymbol{\theta}_q) $ are similar in terms of function. The difference is that a critic takes both $ a $ and $ s $ as input and then predict Q value, but $ Q(s,a;\boldsymbol{\theta}_q) $ estimates all actions' corresponding Q values with input $ s $.

\begin{algorithm}
	\caption{DDPG algorithm.}
	\begin{algorithmic}[1]
		\STATE \emph{Input:} Episode times $ N_e $, exploration times $ T $, actor learning rate $ \eta_a $, critic learning rate $ \eta_c $.
		\STATE \emph{Initialization:} Initialize actor $ A(s;\boldsymbol{\theta}_a) $ and critic $ C(s,a;\boldsymbol{\theta}_c) $ with random parameter $ \boldsymbol{\theta}_a $ and $ \boldsymbol{\theta}_c $.
		\FOR{$ k = 1 $ to $ N_e $}
		\STATE Receive initial state $ s^1 $, obtain $ s^1_c $ with \eqref{equ:sinr} and \eqref{equ:C}.
		\FOR{$ t = 1 $ to $ T $}
		\STATE Get action $ a^t $ by \eqref{equ:opt-ddpg-noise}. 
		\STATE Execute action $ a^t $, achieve reward $ r^t $, observe new state $ s^{t+1} $ and obtain $ s^{t+1}_c $ with \eqref{equ:sinr} and \eqref{equ:C}.
		\STATE Calculate critic gradient $ \nabla \boldsymbol{\theta}_c $ by \eqref{equ:opt-ddpg-delta1}, update parameter along negative gradient direction $ \boldsymbol{\theta}_c \gets \boldsymbol{\theta}_c - \eta_c \nabla {\boldsymbol{\theta}_c} $.
		\STATE Calculate actor gradient $ \nabla \boldsymbol{\theta}_a $ by \eqref{equ:opt-ddpg-delta2}, update parameter along positive gradient direction $ \boldsymbol{\theta}_a \gets \boldsymbol{\theta}_a + \eta_a \nabla {\boldsymbol{\theta}_a} $.
		\STATE $ s^t \gets s^{t+1} $.
		\STATE $ s^t_c \gets s^{t+1}_c $.
		\ENDFOR
		\ENDFOR
		\STATE \emph{Output:} Learned actor $ A(s;\boldsymbol{\theta}_a) $ and critic $ C(s,a;\boldsymbol{\theta}_c) $.
	\end{algorithmic}
	\label{alg:DDPG}
\end{algorithm}

\section{Simulation Results}\label{sec:sim}

\subsection{Simulation Configuration}

In the training procedure, a cellular network with $ N = 25 $ cells is considered. In each cell, $ K = 4 $ APs are located uniformly and randomly within the range $ [R_{\min}, R_{\max}] $, where $ R_{\min} = 0.01 $ km and $ R_{\max} = 1 $ km are the inner space and half cell-to-cell distance, respectively. The Doppler frequency $ f_d = 10 $ Hz and time period $ T_s = 20 $ ms are adopted to simulate the fading effects. According to the LTE standard, the large-scale fading is modeled as
\begin{equation}\label{equ:beta}
\beta = - 120.9 - 37.6 \log_{10}(d) + 10 \log_{10}(z)
\end{equation}
where log-normal random variable $ z $ follows $ \ln z \sim \mathcal{N}(0, \sigma^2_z) $ with $ \sigma_z^2 = 8 $ dB, and $ d $ is the transmitter-to-receiver distance. The additional white Gaussian noise (AWGN) power $ \sigma^2 $ is $ -114 $ dBm, and the emitting power constraints $ P_{\textup{min}} $ and $ P_{\textup{max}} $ are $ 5 $ dBm and $ 38 $ dBm, respectively. Besides, the maximal SINR is restricted by $ 30 $ dB. 
 
The cardinality of adjacent cells is $ |D_n| = 18, \forall n $, the first $ I_c = 16 $ interferers remain and power level number $ |A| = 10 $. Therefore, the input state dimensions $ |S| $ with feature $ \mathrm{f}_1 $, $ \mathrm{f}_2 $ are $ 32 $ and $ 48 $, respectively. The weight coefficient $ \alpha = 1 $. In episode $ k $, the large-scale fading is invariant and thus the number of episode $ N_e = 5000 $, being large to overcome the generalization problem. The time slots per episode $ T = 10 $, being small to reduce over-fitting in $ k $-th specific scenario. The Adam~\cite{Kingma2014Adam} is adopted as the optimizer for all DRLs. In Table.~\ref{tab:tab1}, the architectures of all DNNs and the hyper-parameter settings are listed in detail. The left and right parts of the layer are activation function and neuron number, respectively. These default settings will be clarified once changed in the following simulations. The training procedure is independently repeated $ 50 $ times for each algorithm design, and the testing result is obtained from $ 500 $ times generated scenarios. The simulation codes are available at~\url{https://github.com/mengxiaomao/DRL_PA}.

\begin{table} 
	\centering 
	\caption{Hyper-parameters setup and DNN architecture.} 
	\begin{tabular}{ccccc} 
		\toprule  
		\multirow{3}{*}{Setting} & \multicolumn{4}{c}{Algorithm} \\ 
		\cline{2-5} 
		& \multirow{2}{*}{RF} & \multirow{2}{*}{DQL} & \multicolumn{2}{c}{DDPG} \\
		& & & Actor & Critic \\ 
		\midrule
		Learning rate & $ 1e^{-4} $ & $ 1e^{-3} $ & $ 1e^{-4} $ & $ 1e^{-3} $ \\
		\multirow{2}{*}{Exploration} & \multirow{2}{*}{\eqref{equ:exp-value}} & $ \varepsilon_1=0.2 $ & \multirow{2}*{\eqref{equ:opt-ddpg-noise}} & \multirow{2}{*}{-} \\ 
		& & $ \varepsilon_{N_e}=1e^{-4} $ & & \\
		\cline{2-5}	
		Output layer & softmax,$ |A| $ & linear,$ |A| $ & \eqref{equ:ddpg_output},$ 1 $ &  linear,$ 1 $ \\ 
		\multirow{2}{*}{Hidden layer} & ReLU,$ 64 $ & ReLU,$ 64 $ & ReLU,$ 64 $ & \multirow{2}{*}{ReLU,$ 64 $} \\ 
		& ReLU,$ 128 $ & ReLU,$ 128 $ & ReLU,$ 128 $ & \\
		Input layer & linear,$ |S| $ & linear,$ |S| $ & linear,$ |S| $ &  linear,$ I_c $ \\
		\bottomrule 
	\end{tabular} 
\label{tab:tab1} 
\end{table}

\subsection{DRL Algorithm Comparison}\label{A}

In this subsection, the sum-rate performance of \emph{REINFORCE}, DQL and DDPG is studied, in terms of experience replay, feature selection and quantization error. The notations $ \sigma^2_c $, $ \bar{C} $ and $ \bar{C}^* $ are defined as variance of sum-rate, average sum-rate, the average sum-rate of top $ 20\% $ over independent repetitive experiments, respectively. The $ \bar{C}^* $ is an indicator to measure performance of the well-trained algorithms. 

\subsubsection{Experience Replay}\label{ep}

Since the parameter initialization and data generation are stochastic, the performance of DRL algorithms can be influenced to varying degrees. As shown in Table.~\ref{tab:tab2}\footnote[1]{The proposed DDPG is not applicable for experience replay and thus the corresponding simulation result is omitted.}, the \emph{REINFORCE} and experience replay are abbreviated as RF and ER, respectively. Generally, the experience replay helps the DRLs reduce the variance of sum-rate $ \sigma^2_c $ and improve the average sum-rate $ \bar{C} $, but its influence on best results $ \bar{C}^* $ is negligible.

The variance $ \sigma^2_c $ of \emph{REINFORCE} is the highest, and we find it difficult to stabilize the training results even with experience replay and normalization in \eqref{equ:opt-policy-normalization}. In contrast, the DQL is much more stable. While the $ \sigma^2_c $ of DDPG is the lowest, up to one or more orders of magnitude lower than the \emph{REINFORCE}. This indicates that DDPG has strong robustness to random issues. Moreover, DDPG achieves the highest $ \bar{C}^* $. In general, the $ \bar{C}^* $ performance of \emph{REINFORCE} and DQL are almost the same, and \emph{REINFORCE} performs slightly better than DQL but has weaker stabilization. The DDPG overwhelms these two algorithms, in terms of both sum-rate performance and robustness.

\begin{figure}
	\centering
	\includegraphics[width=3.6in]{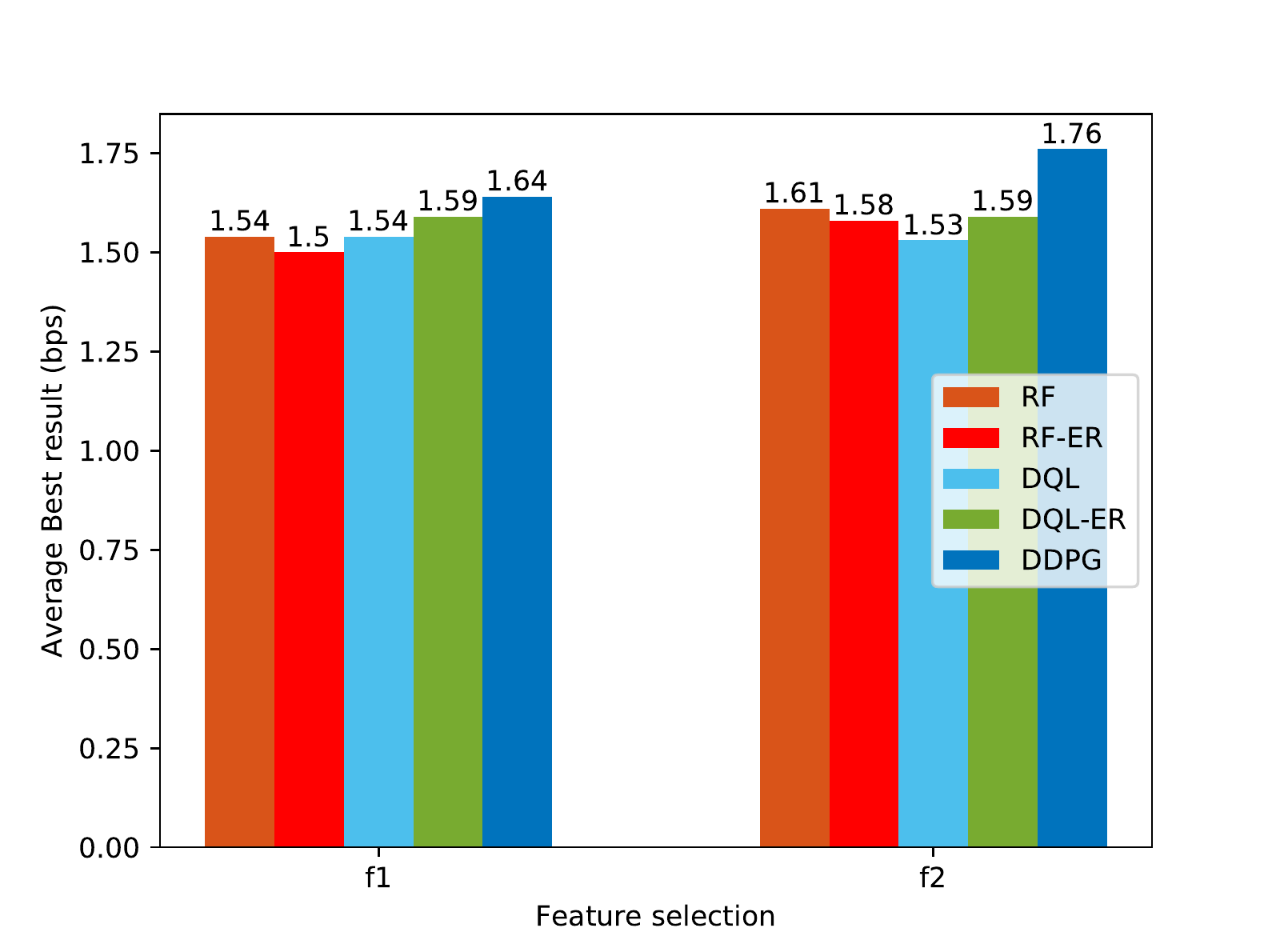}
	\caption{Average best sum-rate $ \bar{C}^* $ of different DRL algorithms, with or without the experience replay technique.}
	\label{fig:exp}
\end{figure}

\subsubsection{Feature Engineering}\label{fe}

Next we compare the performance with features $ \mathrm{f}_1 $ or $ \mathrm{f}_2 $. As shown in Table.~\ref{tab:tab2} and Fig.~\ref{fig:exp}, the assisted information $ \widetilde{\boldsymbol{C}}^{t-1}_{n,k} $ in $ \mathrm{f}_2 $ generally improves the average sum-rate $ \bar{C} $. Besides, the improvement on best results $ \bar{C}^* $ is notable, especially for DDPG algorithm. We speculate that the mapping function is hard to approximate for a simple NN, due to the multiplication, division and exponentiation operations in \eqref{equ:sinr} and \eqref{equ:C}. Meanwhile, the variance $ \sigma^2_c $ is increased with the additional feature. To achieve the highest sum-rate score by repetitive training, this feature is important. The improved performance is achieved at a cost of enlarged input dimension and more training times. On the other hand, a simplified feature state is meaningful for on-line training since the data and computational resource can be restricted and costly. 

\subsubsection{Quantization Error}\label{qe}

In a common sense, the quantization error can be gradually reduced by increasing the digitalizing bit. Therefore, the number of power level $ |A| \in \{3, 6, 10, 14, 20, 40\} $ in this designed experiment, and the best result $ \bar{C}^* $ is used as the measurement. As illustrated in Fig.~\ref{fig:quan}, the $ \bar{C}^* $ of \emph{REINFORCE} and DQL both slightly rise as the $ |A| $ increases from $ 3 $ to $ 10 $. However, further increase of output dimension cannot improve the sum-rate performance. The $ \bar{C}^* $ of DQL drops slowly, while that of \emph{REINFORCE} experiences a dramatic decline from $ 1.54 $ bps to $ 1.19 $ bps, as the $ |A| $ increases from $ 14 $ to $ 40 $. This indicates that the huge action space can lead to difficulties in practical training especially for \emph{REINFORCE}, and also full elimination of quantization error is infeasible by simply enlarging action space. In addition, DDPG has no need for discretization of space by nature, and it outperforms both DQL and \emph{REINFORCE} algorithms.

\begin{figure}
	\centering
	\includegraphics[width=3.6in]{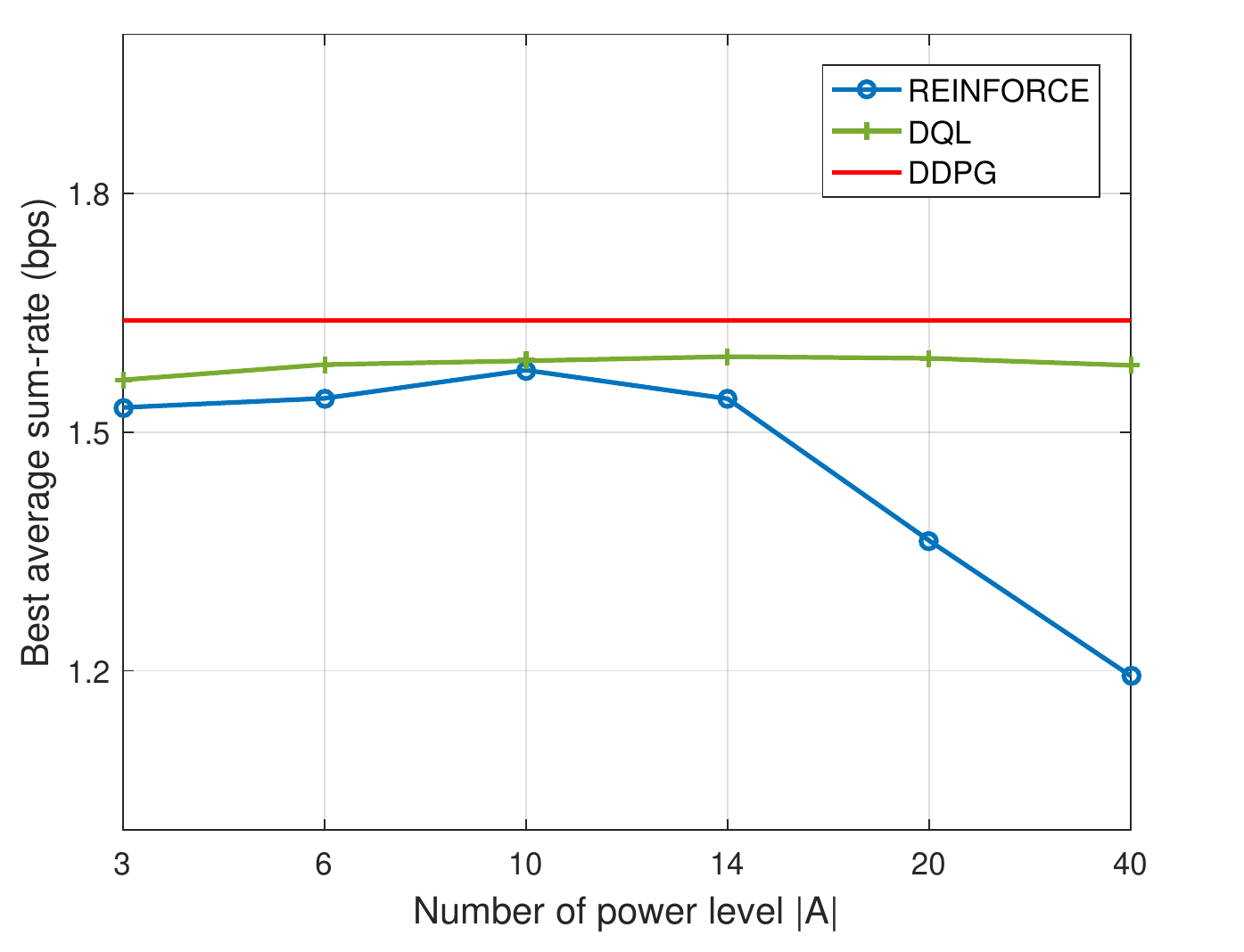}
	\caption{Plot of best result $ \bar{C}^* $ versus power level number $ |A| $ (using feature $ \textup{f}_1 $).}
	\label{fig:quan}
\end{figure}

\begin{table} 
	\centering 
	\caption{Comparison of different DRL design.} 
	\begin{tabular}{cccccc} 
		\toprule  
		\multirow{2}{*}{Variable} & \multicolumn{5}{c}{$ \mathrm{f}_1 $} \\ 
		\cline{2-6} 
		& RF & RF-ER & DQL & DQL-ER & DDPG \\ 
		\midrule
		$ \sigma^2_c $ & $ 3.95e^{-2} $ & $ 3.89e^{-2} $ & $ 2.73^{-2} $ & $ 2.76^{-3} $ & $ 3.37e^{-4} $\\
		$ \bar{C} $ & $ 1.33 $ & $ 1.24 $ & $ 1.44 $ & $ 1.53 $ & $ 1.62 $\\
		$ \bar{C}^* $ & $ 1.54 $ & $ 1.50 $ & $ 1.54 $ & $ 1.59 $ & $ 1.64 $\\
		\cline{1-6}
		& \multicolumn{5}{c}{$ \mathrm{f}_2 $} \\ 
		\midrule
		$ \sigma^2_c $ & $ 8.02e^{-2} $ & $ 5.85e^{-2} $ & $ 1.49e^{-2} $ & $ 8.99e^{-3} $ & $ 2.48e^{-3} $\\
		$ \bar{C} $ & $ 1.32 $ & $ 1.38 $ & $ 1.39 $ & $ 1.50 $ & $ 1.71 $\\
		$ \bar{C}^* $ & $ 1.61 $ & $ 1.58 $ & $ 1.53 $ & $ 1.59 $ & $ 1.76 $\\
		\bottomrule 
	\end{tabular} 
	\label{tab:tab2} 
\end{table}

\subsection{Generalization Performance}\label{vs}

For the following simulations, the learned models with the best result $ \bar{C}^* $ and feature $ \mathrm{f}_2 $ are selected for further study. In the previous subsection, we mainly focus on comparisons between different DRL algorithms, and the training set and testing set are I.I.D. However, the statistical characteristics in real scenarios vary over time, and tracking the environment with frequent on-line training is impractical. Therefore, a good generalization ability is significant to be robust against changing issues. The FP, WMMSE, maximum power and random power schemes are considered as benchmarks to evaluate our proposed DRL algorithms.  

\subsubsection{Cell range}\label{range}

In this part, the half cell-to-cell range $ R_{\max} $ is regarded as a variable. Nowadays, the cells are getting smaller, and thus the range set $ \{0.1,0.2,0.3,0.4,0.6,0.8,1.0,1.2,1.5\} $ km is considered. As shown in Fig.~\ref{fig:dis}, generally the intra/inter cell interference is stronger as the cell range becomes smaller, and thus the average sum-rate decreases. The sum-rate performance of random/maximum power is the lowest, while the FP and WMMSE achieve much higher spectral efficiency. The performances of these two algorithms are comparable, and WMMSE performs slightly better than FP. In contrast, all the data-driven algorithms outperform the model-driven methods, and the proposed actor-critic-based DDPG achieves the highest sum-rate value. Additionally, the learned models are obtained in the simulation environment with fixed range $ R_{\max} $, but performance degradation in these unknown scenarios is not found. Therefore, our learned data-driven models with proposed algorithms show good generalization ability in terms of varying cell ranges $ R_{\max} $.

\begin{figure}
	\centering
	\includegraphics[width=3.6in]{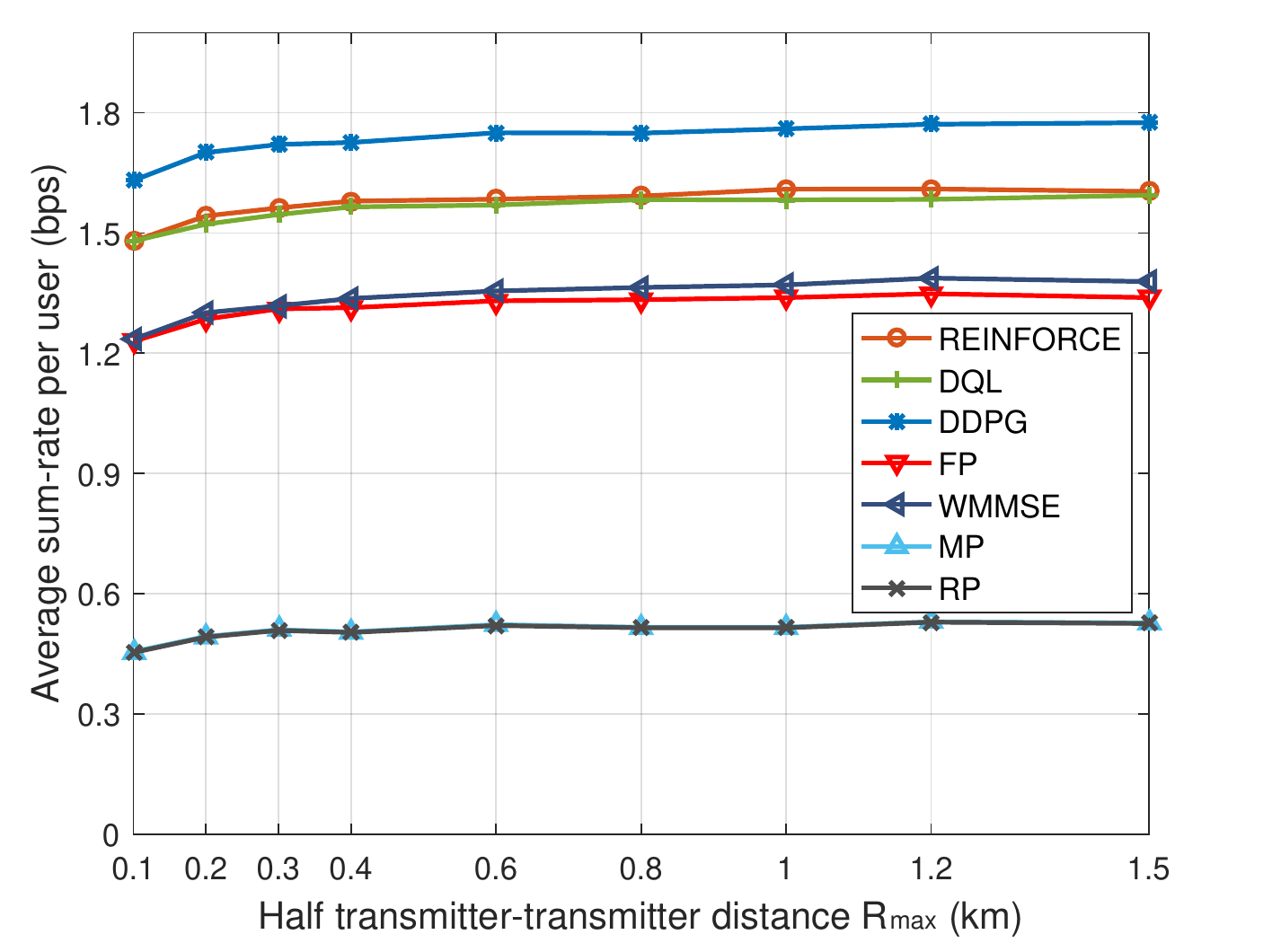}
	\caption{Average sum-rate versus half cell range $ R_{\max} $.}
	\label{fig:dis}
\end{figure}

\subsubsection{User Density}\label{density}

In a practical scenario, the user density can change over time and location, so it is considered in this simulation. The user density is changed by the number of AP per cell $ K $, which ranges from $ 1 $ to $ 8 $. As plotted in Fig.~\ref{fig:maxM}, the average sum-rate drops as the users become denser, and all the algorithms have the similar trend. Apparently, the DRL approaches outperform the other schemes, and DDPG again achieves the best sum-rate performance. Hence, the simulation result shows that the learned data-driven models also show good generalization ability on different user densities.

\begin{figure}
	\centering
	\includegraphics[width=3.6in]{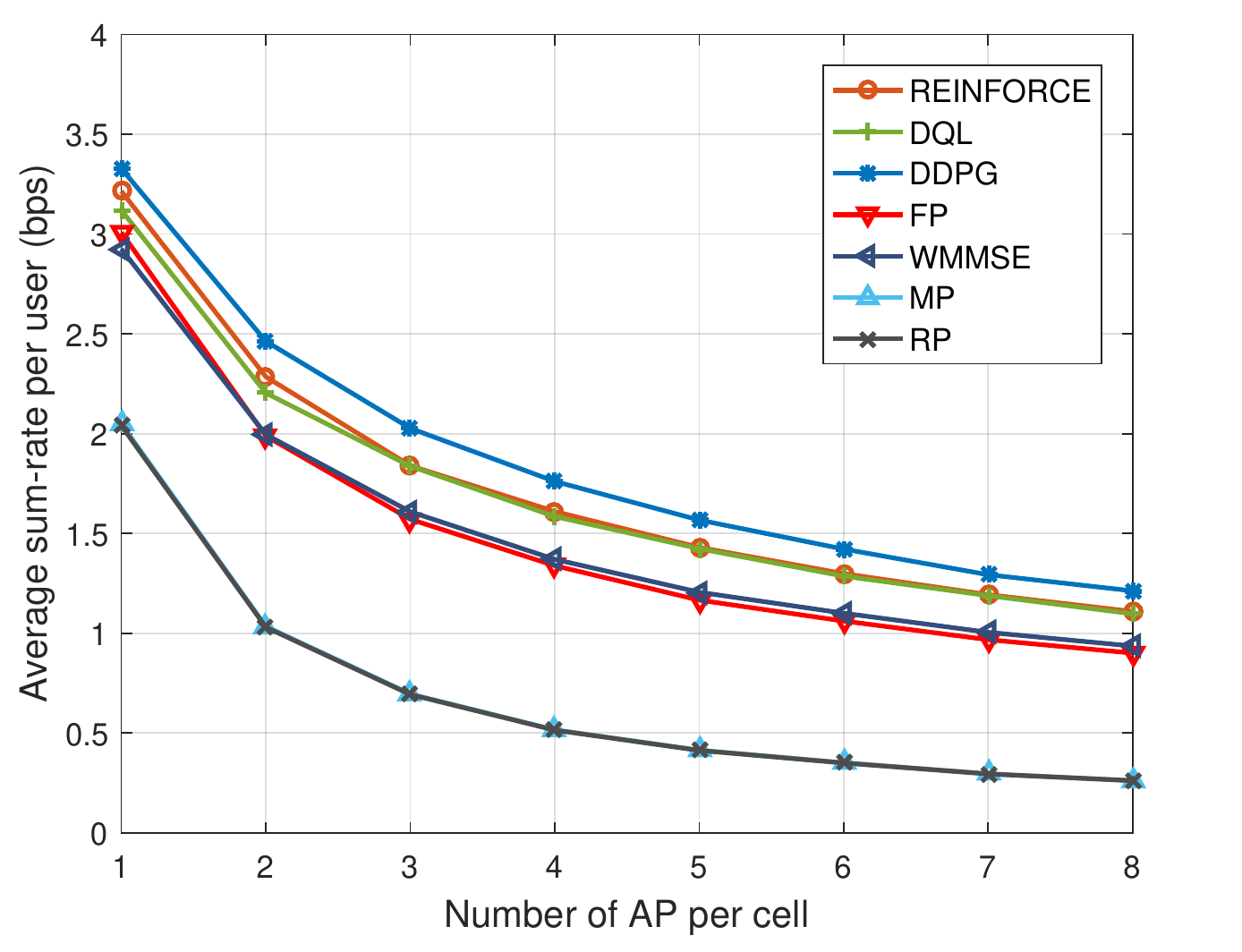}
	\caption{Average sum-rate per AP versus different user densities.}
	\label{fig:maxM}
\end{figure}

\subsubsection{Doppler frequency}\label{fd}

The Doppler frequency $ f_d $ is a significant variable related to the small-scale fading. Since the information at last instant is utilized for the current power allocation, fast fading can lead to performance degradation for our proposed data-driven models. Meanwhile, the model-driven algorithms are not influenced by $ f_d $ by nature. The Doppler frequency $ f_d $ is sampled in range from $ 4 $ Hz to $ 18 $ Hz, and the simulation results in Fig.~\ref{fig:fd} show that the average sum-rates of data-driven algorithms drop slowly in this $ f_d $ range. This indicates that the data-driven models also are robust against Doppler frequency $ f_d $.

\begin{figure}
	\centering
	\includegraphics[width=3.6in]{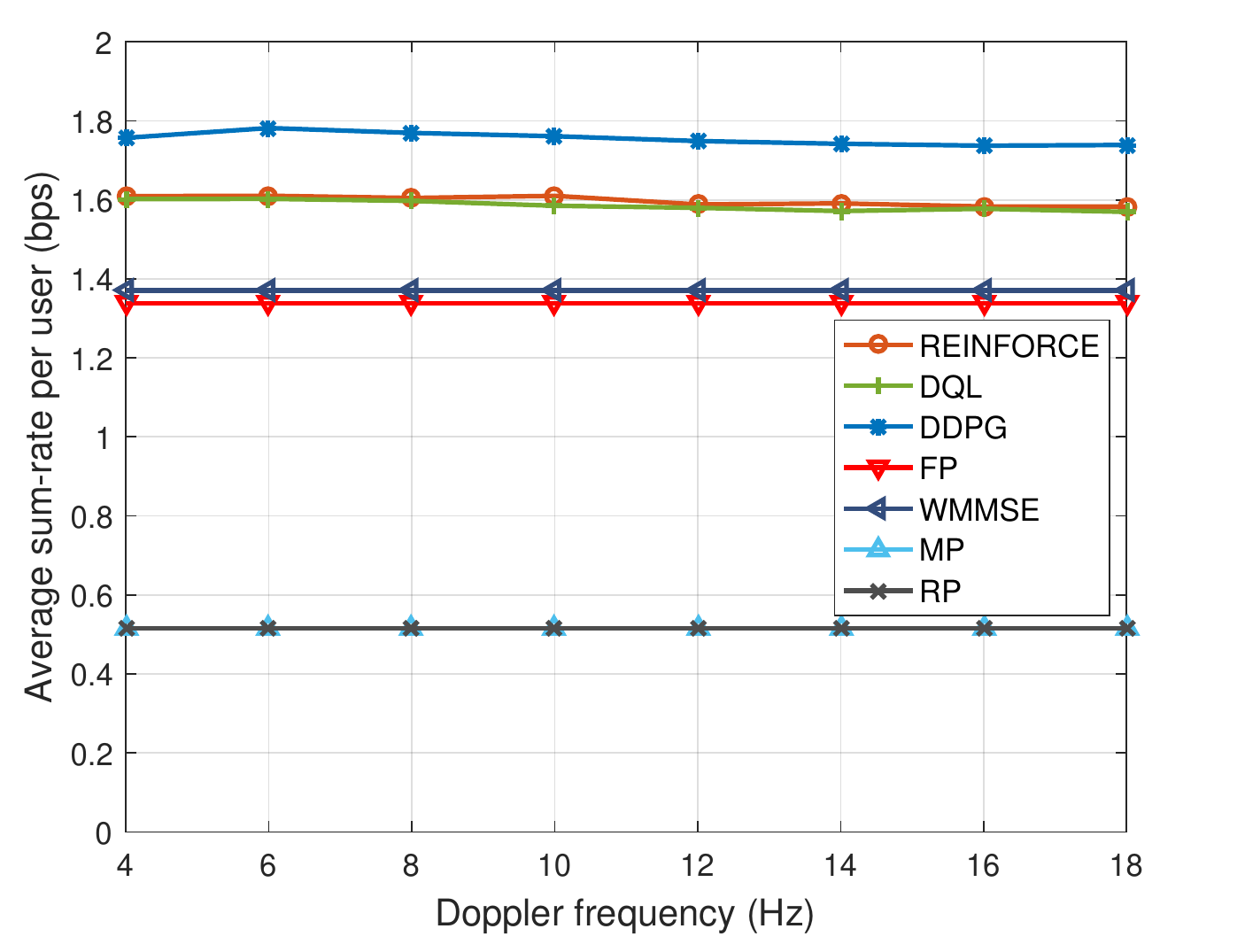}
	\caption{Average sum-rate versus different Doppler frequency $ f_d $.}
	\label{fig:fd}
\end{figure}

\subsection{Computation Complexity}\label{cc}

Low computation complexity is crucial for algorithm deployment and it is considered here. The simulation platform is presented as: CPU Intel i7-6700 and GPU Nvidia GTX-1070Ti. There are $ 100 $ APs in the simulated cellular network, the time cost per execution $ T_c $ of our proposed distributed algorithms and the centralized model-based methods are listed in Table.~\ref{tab:tab3}. It is interesting that the calculation time with GPU is higher than that of CPU, and we consider that the GPU cannot be fully utilized with small scale DNN and distributed execution\footnote[2]{The common batch operation cannot be used under distributed execution in real scenario.}. It can be seen that the time cost of three DRLs are almost the same due to similar DNN models, and in terms of only CPU time, they are about $ 15.5 $ and $ 61.0 $ times faster than FP and WMMSE, respectively. Fast execution speed with DNN tools can be explained by several points:
\begin{enumerate}
	\item The execution of our proposed algorithms is distributed, and thus the time expense is a constant as the total amount of users $ NK $ increases, at a cost of more calculation devices (equal to $ NK $).
	\item Most of the operations in DNNs involve matrix multiplication and addition, which can be accelerated by parallel computation. Besides, the simple but efficient activation function ReLU: $ \max(\cdot,0) $ is adopted.
\end{enumerate}
In summary, the low computational time cost of the proposed DRLs can be attributed to distributed execution framework, parallel computing architecture and simple efficient function.  

\begin{table} 
	\centering 
	\caption{Average time cost per execution $ T_c $ ($ \textup{sec} $).} 
	\begin{tabular}{cccccc} 
		\toprule  
		\multirow{2}{*}{} & \multicolumn{5}{c}{Algorithm} \\ 
		\cline{2-6} 
		& RF & DQL & DDPG & FP & WMMSE \\ 
		\midrule
		CPU & $ 3.21e^{-4} $ & $ 3.10e^{-4} $ & $ 3.11e^{-4} $ & $ 4.80^{-3} $ & $ 1.84e^{-2} $\\
		GPU & $ 4.48e^{-4} $ & $ 4.47e^{-4} $ & $ 4.63e^{-4} $ & - & -\\
		\bottomrule 
	\end{tabular} 
	\label{tab:tab3} 
\end{table}

\section{Conclusions \& Discussions}\label{sec:con}

The distributed power allocation with proposed DRL algorithms in wireless cellular networks with IMAC was investigated. We presented a mathematically analysis on the proper design and application of DRL algorithms at a systematic level by considering inter-cell cooperation, off-line/on-line training and distributed execution. The concrete algorithm design was further introduced. In theory, the sum-rate performances of DQL and \emph{REINFORCE} algorithms are the same with proper training, and DDPG outperforms these two methods by eliminating quantization error. The simulation results agree with our expectation, and DDPG performs the best in terms of both sum-rate performance and robustness. Besides, all the data-driven approaches outperform the state-of-art model-based methods, and also show good generalization ability and low computational time cost in a series of experiments. 

The data-driven algorithm, especially DRL, is a promising technique for future intelligent networks, and the proposed DDPG algorithm can be applied to general tasks with discrete/continuous state/action space and joint optimization problems of multiple variables. Specifically speaking, the algorithm can be applied to many problems such as user scheduling, channel management and power allocation in various communication networks. 

\section{Acknowledgments}

This work was supported in part by the National Natural Science Foundation of China (Grant No. 61801112, 61601281), the Natural Science Foundation of Jiangsu Province (Grant No. BK20180357), the Open Program of State Key Laboratory of Millimeter Waves (Southeast University, Grant No. Z201804).


\end{document}